\documentclass[a4paper]{article}
\usepackage{url}
\usepackage{graphicx}
\usepackage{amssymb,stmaryrd}
\usepackage{float,algorithm,algorithmic} 
\usepackage{amsthm}
\usepackage{hyperref}

\newtheorem{lemma}{Lemma}
\newcommand{\graphwidth}{\textwidth}

\title{On using floating-point computations to help an exact linear arithmetic decision procedure%
\thanks{This work was partially funded by the \href{http://www.agence-nationale-recherche.fr/}{ANR} \href{http://www.agence-nationale-recherche.fr/?NodId=17&lngAAPId=159}{ARPEGE} project ``\href{http://asopt.inrialpes.fr/index.php/Main_Page}{ASOPT}''.}}

\author{\href{http://www-verimag.imag.fr/~monniaux/}{David Monniaux}\\
CNRS / VERIMAG%
\thanks{\href{http://www-verimag.imag.fr/}{VERIMAG} is a joint laboratory of \href{http://www.cnrs.fr/}{CNRS}, \href{http://www.ujf-grenoble.fr/}{Universit\'e Joseph Fourier} and \href{http://www.grenoble-inp.fr/}{Grenoble-INP}.}}
\begin{document}
\maketitle

\newcommand{\soft}[1]{\textsc{#1}}
\newcommand{\ve}[1]{\mathbf{#1}}
\newcommand{\algo}[1]{\textsc{#1}}

\begin{abstract}
We consider the decision problem for quantifier-free formulas whose atoms are linear inequalities interpreted over the reals or rationals. This problem may be decided using satisfiability modulo theory (SMT), using a mixture of a SAT solver and a simplex-based decision procedure for conjunctions. State-of-the-art SMT solvers use simplex implementations over rational numbers, which perform well for typical problems arising from model-checking and program analysis (sparse inequalities, small coefficients) but are slow for other applications (denser problems, larger coefficients).

We propose a simple preprocessing phase that can be adapted to existing SMT solvers and that may be optionally triggered. Despite using floating-point computations, our method is sound and complete --- it merely affects efficiency. We implemented the method and provide benchmarks showing that this change brings a naive and slow decision procedure (``textbook simplex'' with rational numbers) up to the efficiency of recent SMT solvers, over test cases arising from model-checking, and makes it definitely faster than state-of-the-art SMT solvers on dense examples. 
\end{abstract}
\section{Introduction}
Decision procedures for arithmetic theories are widely used for computer-aided verification. A decision procedure for a theory $T$ takes as input a formula of $T$ and outputs a Boolean: whether the formula is satisfiable. For many decidable and potentially useful theories, however, decision procedures are sometimes too slow to process problems beyond small examples. This is for instance the case of the theory of real closed fields (polynomial arithmetic over the real numbers). Excessive computation times arise from two sources: the Boolean structure of the formulas to be decided (propositional satisfiability is currently solved in exponential time in the worst case), and the intrinsic hardness of the theory. In recent years, SAT modulo theory (SMT) techniques have addressed the former source of inefficiency, by leveraging the power of efficient SAT (Boolean satisfiability) solvers to deal with the Boolean structure. SMT solvers combine a SAT solver with a decision procedure for conjunctions of atoms in $T$. If $T$ is linear real arithmetic (LRA), then this decision procedure must decide whether a set of linear inequalities with rational or integer coefficients has rational solutions.

The problem of testing whether a set of linear inequalities has a solution and, if it has, to find a solution that maximizes some linear combination of the variables is known as linear programming and has been considerably studied in operational research. Very efficient implementations exist, whether commercial or not, and are able to solve very large problems. They are not directly applicable to our problems, however, if only because they operate over floating-point numbers and provide in general no assurance that their result is truthful, despite elaborate precautions taken against numerical instabilities. As a result, the decision procedures for LRA in SMT solvers are implemented with rational arithmetic, which is slower than floating-point, especially if coefficients become large, as often happens with dense linear problems: large coefficients force the use of costly extended precision arithmetic. It thus would seem desirable to leverage the speed and maturity of floating-point linear programming systems to enhance exact decision procedures.

This article describes a simple preprocessing phase that can be added, with minimal change, to existing rational simplex implementations used as decision procedures inside SMT solvers. The procedure was implemented on top of a naive and inefficient rational simplex implementation; the resulting procedure rivals recent SMT solvers.

A similar method has been proposed in the operational research field~\cite{Dihflaoui_et_al_SODA03},%
\footnote{We were not aware of this article when we started our work, and we thank Bernd G\"{a}rtner for pointing it to us.}
but there are reasons why it may perform less well for the typical optimization tasks of operational research than for decision tasks. The novelty of this article is the application of this technique as a simple modification of existing SMT algorithms.

\section{Simplex}
SMT solvers need a decision procedure capable of:
\begin{itemize}
\item being used incrementally: adding new constraints to the problem, and removing blocks of constraints, preferably without recomputing everything;
\item telling whether the problem is satisfiable or not;
\item if the problem is unsatisfiable, outputting a (preferably small or even minimal) unsatisfiable subset of the constraints;
\item propagating theory lemmas, if possible at reasonable costs (from a conjunction $C_1 \land \dots \land C_n$, obtain literals $L_1, \dots, L_m$ that are consequences of that conjunction: $C_1 \land \dots \land C_n \Rightarrow L_1 \land \dots \land L_m$).
\end{itemize}

All current SMT solvers seem to decide general linear real arithmetic (as opposed to syntactic restrictions thereof such as difference logic) using the \emph{simplex algorithm}. This algorithm is exponential in the worst case, but tends to perform well in practice; none of the current solvers seem to use a (polynomial-time) interior point method. Our method is a variant of the simplex algorithm; we shall thus first describe the ``conventional'' simplex.

\subsection{Basic simplex}\label{part:basic_simplex}
We shall first give a brief summary on the dual simplex algorithm on which the LRA decision procedures in \soft{Yices}%
\footnote{\url{http://yices.csl.sri.com/}}
\cite{SRI-CSL-06-01,DBLP:conf/cav/DutertreM06} and \soft{Z3}%
\footnote{\url{http://research.microsoft.com/en-us/um/redmond/projects/z3/}}
\cite{DBLP:conf/tacas/MouraB08}
are based. There otherwise exist many excellent textbooks on the simplex algorithm \cite{Dantzig1998,Schrijver98}, though these seldom discuss the specifics of implementations in exact precision or incremental use.

Take a system of linear equations, e.g.
\begin{equation}
\left\{\begin{array}{ll}
x - 2y & \leq 1 \\
-y + 3z & \geq -1 \\
x - 6z & \geq 4
\end{array}
\right.
\end{equation}

The system is first made canonical. Inequalities are scaled so that each left hand side only has integer coefficients with no common factors. Then, each inequality is optionally negated so that the first coefficient appearing (using some arbitrary ordering of the variables) is positive. This ensures that two inequalities constraining the same direction in space (e.g. $-y + 3z \geq -1$ and $2y-6z \geq 3$) appear with the exact same left-hand side. For each left-hand side that is not a variable, a new variable is introduced; the system is then converted into a number of linear equalities and bound constraints on the variables. For instance, the above system gets converted into:
\begin{equation}
\left\{\begin{array}{ll}
\alpha &= x - 2y \\
\beta &= y - 3z \\
\gamma &= x - 6z
\end{array}\right.
\quad
\left\{\begin{array}{ll}
\alpha \leq 1\\
\beta \leq 1\\
\gamma \geq 4
\end{array}\right.
\end{equation}
The problem is thus formulated as deciding whether a product of intervals intersects a linear subspace given by a basis.

The set of variables is partitioned into \emph{basic} and \emph{nonbasic} variables; the number of basic variables stays constant throughout the algorithm. Basic variables are expressed as linear combinations of the nonbasic variables. The main operation of the algorithm is \emph{pivoting}: a basic variable is made nonbasic and a nonbasic variable is made basic, without changing the linear subspace defined by the system of equations. For instance, in the above example, $\alpha = x - 2y$ defines the basic variable in term of the nonbasic variables $x$ and~$y$. If one wants instead $x$ to be made basic, one obtains $x = \alpha + 2y$. The variable $x$ then has to be replaced by $\alpha + 2y$ in all the other equalities, so that the right-hand sides of these equalities only refer to nonbasic variables. This replacement procedure (essentially, replacing a vector $\ve{u}$ by $\ve{v}+k\ve{u}$) is the costly part of the algorithm. A more formal description is given in Alg.~\ref{algo:Pivot}.

Let us insist that pivoting does not change anything to the validity of the problem: both the bounds and the linear subspace stay the same. The idea behind the simplex algorithm is to pivot until a position is found where it is obvious that the problem has a solution, or that it has not.

The algorithm also maintains a vector of ``current'' values for the variables. This vector is fully defined by its projection on the nonbasic variables, since the basic variables can be obtained from them. The current values of the nonbasic variables always stay within their respective bounds. If all current values of the basic variables also fit within the bounds, the current value is a solution to the problem and the algorithm terminates.

If there are basic variables that fall outside the prescribed bounds, one of them (say, $\alpha$) is selected and the corresponding row (say, $x-2y$) is examined. Suppose for the sake of the explanation that the current value for $\alpha$, $c_\alpha$, is strictly greater than the maximal prescribed value $M_\alpha$. One can try making $x$ smaller or $y$ larger to compensate for the difference. If $x$ is already at its lower bound and $y$ at its upper bound, it is impossible to make $\alpha$ smaller and the system is unsatisfiable; the algorithm then terminates. In other words, by performing interval arithmetic over the equation defining $\alpha$ in terms of the nonbasic variables, one shows this equation to be unsatisfiable (replacing the nonbasic variables by their interval bounds, one obtains an interval that does not intersect the interval for~$\alpha$).

Let us now suppose that $x$ is not at its lower bound; we can try making it smaller. $\alpha$ and $x$ are pivoted: $x$ becomes basic and $\alpha$ nonbasic. $\alpha$ is set to its lower bound and $x$ is adjusted accordingly.

It can be shown that if there is a solution to the problem, then there is one configuration where any nonbasic variable is either at its lower or upper bound. Intuitively, if some nonbasic variables are ``in the middle'', then this means we have some ``slack margin'' available, so we should as well use it.
The simplex algorithm appears to move between numerical points taken in an infinite continuous space, but in fact, its current configuration is fully defined by stating which variables are basic and nonbasic, and, for the nonbasic variables, whether they are at their lower or upper bound --- thus the number of configurations is finite.

Remark that we left aside how we choose which basic variable and which basic variable to pivot. It can be shown that certain pivoting strategies (say, choose the suitable nonbasic and basic variables of least index) necessarily lead, maybe after a great number of pivots, to a configuration where either a solution is obtained, or the constraints of a nonbasic variable are clearly unsatisfiable.

The idea of our article is based on the following remark: the simplex algorithm, with a suitable pivoting strategy, always terminates with the right answer, but its execution time can vary considerably depending on the initial configuration. If it is started from a configuration where it is obvious that the system has a solution or does not have one, then it terminates immediately. Otherwise, it may do a great deal of pivoting. Our idea is to use the output of some untrusted floating-point simplex algorithm to direct the rational simplex to a hopefully good starting point.

\subsection{Modifications and extensions}
\label{part:extensions}
The above algorithm is a quick description of ``textbook simplex''. It is not sufficient for obtaining a numerically stable implementation if implemented over floating-point; this is why, after initial trials with a naive floating-point implementation, we moved to a better off-the-shelf implementation, namely \soft{Glpk} (GNU Linear Programming Kit)~\cite{GLPK}.

So far, we have only considered wide inequalities.
A classical method to convert problems with strict inequalities into problems with only wide inequalities is to introduce \emph{infinitesimals}: a coordinate is no longer one single rational number, but a pair of rational numbers $a$ and $b$, denoted $a+b\varepsilon$, with lexicographic ordering. $\varepsilon$ is thus a special number, greater than zero but less than all positive rationals. $x < y$ is rewritten into $x + \varepsilon \leq y$. The ``current'' values in the simplex algorithm are thus pairs $(a,b)$ of rationals, noted $a+b\varepsilon$, the upper bounds are either $+\infty$ or $a+b\varepsilon$, the lower bounds either $-\infty$ either $a+b\varepsilon$. The ``current'' vector, a vector of pairs, can be equivalently represented as a pair of vectors of rationals $(\ve{u},\ve{v})$, written $\ve{u}+\varepsilon \ve{v}$ --- if $\ve{v} \neq \ve{0}$, it designates a point infinitely close to $\ve{u}$ in the direction of $\ve{v}$, if $\ve{v} = \ve{0}$ it only means the point~$\ve{u}$.

In addition to a yes/no answer, decision procedures used in SMT need to provide:
\begin{itemize}
\item In case of a positive answer, a solution point. If only wide inequalities are used, this is given straightforwardly by the ``current'' values. For points of the form $\ve{u}+\varepsilon \ve{v}$, one considers the half-line $\ve{u}+t\ve{v}$, $t > 0$, inject it into the system of inequalities, and solve it for $t$ --- the solution set will be an interval of the form $(0,t_0)$. A suitable solution point is thus $\ve{u}+\frac{t_0}{2} \ve{v}$.

\item In case of a negative answer, a contradiction witness: nonnegative coefficients such that by multiplying the original inequalities by those coefficients, one gets a trivially unsatisfiable inequality ($0 < c$ where $c \leq 0$, or $0 \leq c$ where $c < 0$). This contradiction witness is obtained by using an auxiliary tableau tracking how the equalities $b - \sum_n t_{b,n} n = 0$ defining the basic variables were obtained as linear combinations of the original equalities defined at the initialization of the simplex, as described in \cite[Sec.~3.2.2]{SRI-CSL-06-01}.
\end{itemize}

\section{Mixed floating-point / rational strategy}
\label{part:mixed}
Our procedure takes as input a rational simplex problem in the format described at Sec.~\ref{part:basic_simplex}: a tableau of linear equalities and bounds on the variables. It initializes a floating-point simplex by converting the rational problem: the infinitesimals are discarded and the rationals rounded to nearest. It then calls, as a subroutine, a floating-point simplex algorithm which, on exit, indicates whether the problem is satisfiable (at least according to floating-point computation), and a corresponding configuration (a partition into basic and nonbasic variables, and, for each nonbasic variable, whether it is set to its upper or lower bound). There are several suitable packages available; in order to perform experiments, we implemented our method using the GNU Linear programming toolkit (\soft{Glpk})~\cite{GLPK}.

In order for the resulting mixed procedure to be used incrementally, the floating-point solver should support incremental use. Commercial linear programming solvers are designed for large problems specified as a whole; there may be a large overhead for loading the problem into the solver, even if the problem is small. Instead, we need solvers capable of incrementally adding and withdrawing constraint bounds at minimal cost. \soft{Glpk} supports incremental use, since it keeps the factorization of the basis in memory between calls \cite[p.~20]{GLPK}; this factorization is costly to compute but needs to be computed only if the basis matrix changes: in our case, this basis stays the same.

At this point, if successful, and unless there has been some fatal numeric degeneracy, the floating-point simplex outputs a floating-point approximation to a solution point. However, in general, this approximation, converted back into a rational point, is not necessarily a true solution point. The reason is that simplex produces solution points at a vertex of the solution polyhedron, and, numerically speaking, it is in general impossible to be exactly on that point; in general, the solution point obtained is found to be very slightly outside of the polyhedron when membership is tested in exact arithmetic. It is therefore not feasible to simply take this solution point as a witness.

The rational simplex tableau is then ``forcefully pivoted'', using Algorithm~\ref{algo:ForcedPivot}, until it has the same basic/nonbasic partition as the floating-point output. This amounts to a pass of Gaussian elimination for changing the basis of the linear subspace. This phase can partially fail if the final basic/nonbasic partition requested is infeasible in exact precision arithmetic --- maybe because of bugs in the floating-point simplex algorithm, or simply because of floating-point inaccuracies.

The ``current'' values of the nonbasic variables are then initialized according to the output of the floating-point simplex: if the floating-point simplex selected the upper bound for nonbasic variable $n$ then its current value in the rational simplex is set to its upper bound, and similarly for lower bounds. If the final basic/nonbasic partition produced by the floating-point simplex is infeasible, then there are nonbasic variables of the rational simplex for which no information is known: these are left untouched or set to arbitrary values within their bounds (this does not affect correctness).  The current values of the basic variables are computed using the rational tableau.

The rational simplex is then started. If things have gone well, it terminates immediately by noticing that it has found either a solution, or a configuration showing that the system is unsatisfiable. If things have gone moderately well, the rational simplex does a few additional rounds of pivoting. If things have gone badly, the rational simplex performs a full search.

The rational simplex algorithms are well known, and we have already presented them in Sect.~\ref{part:basic_simplex}. The correctness of our mixed algorithm relies on the correctness of the rational simplex and 
the ``forced pivoting'' phase maintaining the invariant that the linear equalities define the same solutions as those in the initial system.

\begin{algorithm}
\caption{$\algo{Pivot}(\textit{tableau}, b, n)$: pivot the basic variable $b$ and the nonbasic variable $n$. $t_v$ is the line defining basic variable $v$, $t_{v,w}$ is the coefficient of $t_v$ corresponding to the nonbasic variable $w$. The $a_i$ are the optional auxiliary tableau described in~Sec.~\ref{part:extensions}.}
\label{algo:Pivot}
\begin{algorithmic}
\REQUIRE{$b$ basic, $n$ nonbasic, $t_{b,n} \neq 0$}
\STATE{$p := -1/t_{b,n}; t_n := p.t_b$; $t_{n,n} := 0$; $t_{n,b} := -p$; $a_n = p.a_b$}
\STATE{$B := B \cup \{ n \} \setminus \{ b \}$}
\FORALL{$b' \in B$}
  \STATE{$p := t_{b',n}$; $t_{b',n} := 0$; $t_{b'} := t_{b'} + p.t_b$; $a_{b'} := a_{b'} + p.a_b$}
\ENDFOR
\STATE{$t_b := 0$}
\end{algorithmic}

The ``for all'' loop is the most expensive part of the whole simplex algorithm.
Note that, depending on the way the sparse arrays and auxiliary structures are implemented, this loop may be parallelized, each iteration being independent of the others. This gives a performance boost on dense matrices.
\end{algorithm}

\begin{algorithm}
\caption{$\algo{ForcedPivot}(\textit{tableau}, B_f)$: force pivoting until the set of basic variables is $B_f$}
\label{algo:ForcedPivot}
\begin{algorithmic}
\STATE $B := B_i$
\REPEAT
  \STATE $\textit{hasPivotedSomething} := \FALSE$
  \FORALL{$b \in B \setminus B_f$}
    \IF{$\exists n \in B_f \setminus B~ t_{b,n} \neq 0$}
      \STATE Choose $n$ in $B_f \setminus B$ such that $t_{b,n} \neq 0$
      \STATE $\algo{Pivot}(\textit{tableau}, b, n)$ \COMMENT{This involves
        $B:=B\cup\{n\}\setminus\{b\}$}
      \STATE{$\textit{hasPivotedSomething} := \TRUE$}
    \ENDIF
  \ENDFOR
\UNTIL{$\neg \textit{hasPivotedSomething}$}
\RETURN {$B = B_f$}
\end{algorithmic}
\end{algorithm}

\begin{figure}[htb]
\begin{center}
\includegraphics[width=\graphwidth]{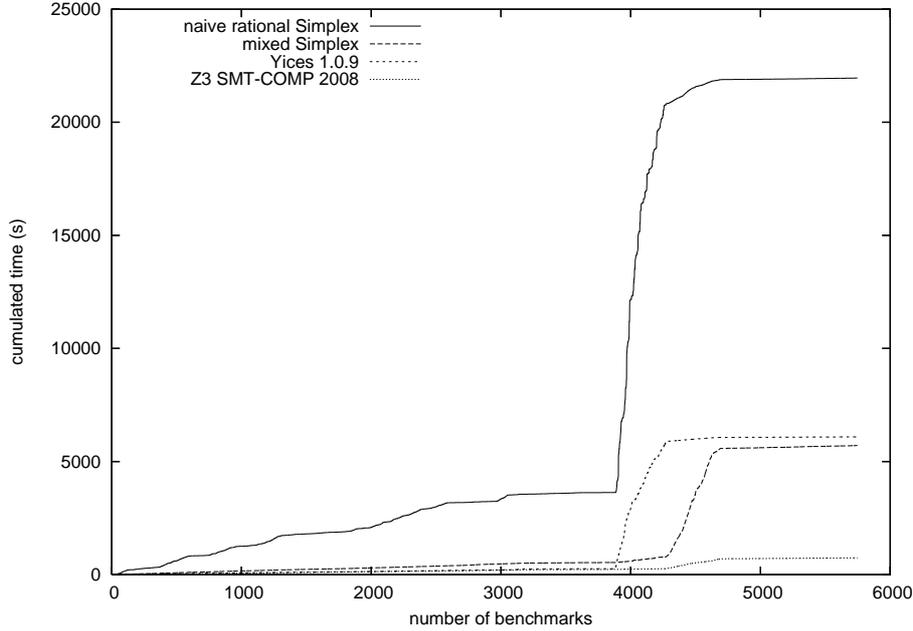}
\end{center}
\caption{\footnotesize Benchmarks on unsatisfiable conjunctions extracted from vSMT verification problems. Even though our implementation of sparse arithmetic and the rational simplex are not up to those in state-of-the-art solvers (as shown by the lack of performance on the ``easy'' examples on the left), and our procedure is not geared towards the sparse problems typical of verification applications, it still performs faster than Yices~1. In 4416 cases out of 5736 (77\%), no additional simplex pivots are needed after \algo{ForcePivot}.
}
\label{fig:benchmarks_Leonardo}
\end{figure}

\begin{figure}[htb]
\begin{center}
\includegraphics[width=\graphwidth]{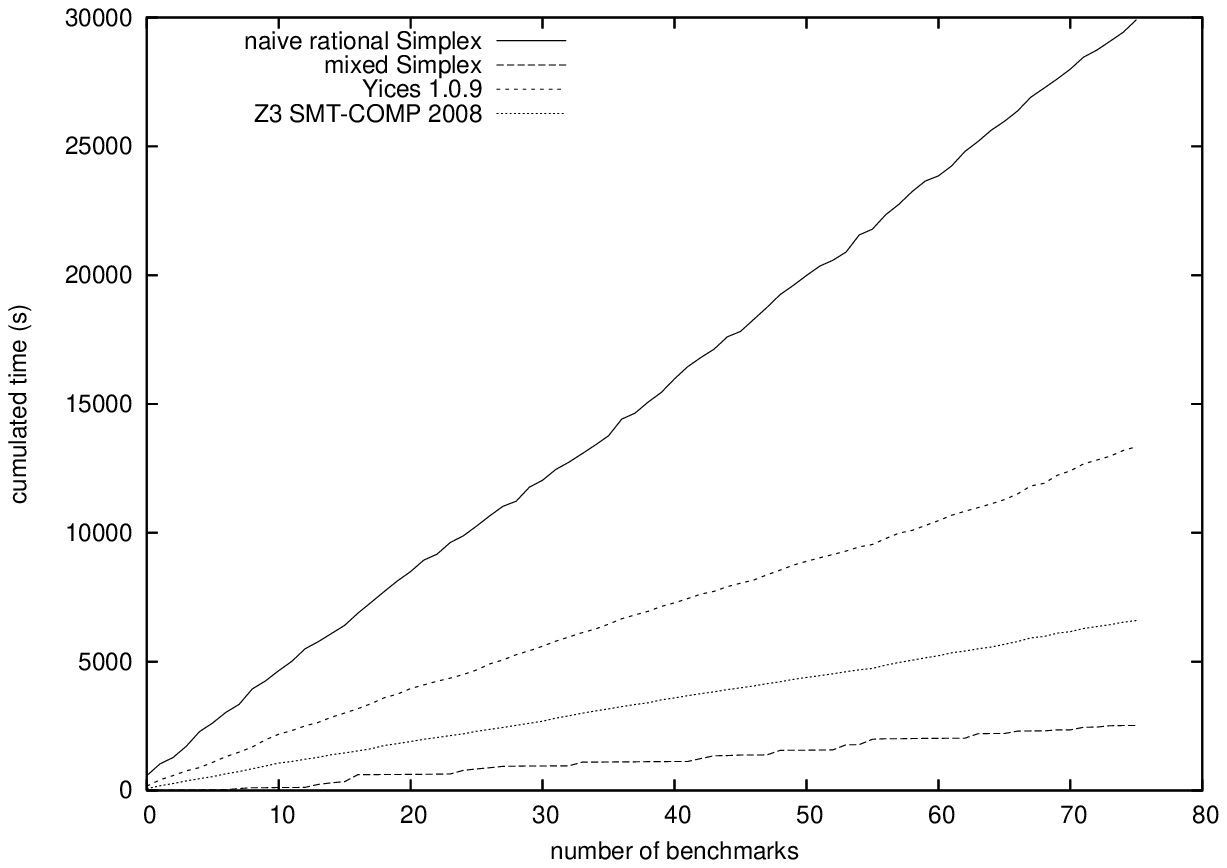}
\end{center}
\caption{\footnotesize Benchmarks on systems of 100 inequalities with 50 variables, where all coefficients are taken uniformly and independently distributed integers in $[-100,100]$. 31 are unsatisfiable, 46 are satisfiable. For each solver, benchmarks sorted by increasing time spent. In 58 cases out of 82 (71\%), no additional simplex pivots are needed after \algo{ForcePivot}.}
\label{fig:benchmarks_random_50_100_100}
\end{figure}

We shall now describe in more detail the ``forced pivoting'' algorithm (Alg.~\ref{algo:ForcedPivot}). This algorithm takes as input a simplex tableau, with associated partition of the set $V$ of variables into basic ($B_i$) and nonbasic variables ($\bar{B}_i$), and a final partition of basic ($B_f$) and nonbasic variables ($\bar{B}_f$). For each basic variable $b$, the tableau contains a line $b = \sum_{n \in \bar{B}} t_{b,n} n$.
Additional constraints are that the tableau is well-formed (basic variables are combination of only the nonbasic variables) and that $|B_i|=|B_f|$ (since $|B|$ is a constant).

Assuming that all arithmetic operations take unit time (which is not true when one considers dense problems, since coefficient sizes quickly grow, but is almost true for sparse problems with small coefficients), the running time of the forced pivoting algorithm is at most cubic in the size of the problem. This motivates our suggestion: instead of performing an expensive and potentially exponential-time search directly with rational arithmetic, we perform it in floating-point, with all the possibilities of optimization of floating-point linear arithmetic offered by modern libraries and compilers, and then perform a cubic-time pass with rational arithmetic.

Not all final configurations are feasible: it is possible to ask \algo{ForcedPivot} to perform an impossible transformation, in which case it returns ``false''. For instance, if the input system of equations is $x = a + b \land y = a + b$,
thus with $B_i = \{ x, y \}$, then it is impossible to move to $B_f = \{ a, b \}$, for there is no way to express $a$ and $b$ as linear functions of $x$ and~$y$. More precisely, we say that a configuration is feasible if it is possible to write the basic variables of the configuration as linear function of the nonbasic variables and obtain the same solutions as the initial system.

\begin{lemma}
\algo{ForcedPivot} succeeds (and returns ``true'') if and only if the final partition defined by $B_f$ is feasible, otherwise it returns ``false''.
\end{lemma}

\begin{proof}
Let $S$ denote the space of solutions of the input system of equations. At all iterations, $S$ is exactly defined by the system of equations, and $\dim S = |\bar{B}|$.
The only way the procedure fails is when $B \neq B_f$ and yet, for all basic variable $b \in B \setminus B_f$ and nonbasic variable $n \in B_f \setminus B$, it is impossible to pivot $b$ and $n$ because $t_{b,n} = 0$. In other words, all such $b$ are linear combinations of the nonbasic variables in $\bar{B} \cap \bar{B}_f$. All variables in $\bar{B}_f$ are thus linear combinations of variables in $\bar{B} \cap \bar{B}_f$, and since we have supposed that $B \neq B_f$, $\bar{B} \cap \bar{B}_f \subsetneq \bar{B}$ thus $|\bar{B} \cap \bar{B}_f| < |\bar{B}| = \dim S$. But then, $|\bar{B}_f| < \dim S$ and $B_f$ cannot be a feasible configuration.
\end{proof}

One can still add a few easy improvements:
\begin{itemize}
\item Before embarking into any simplex, we first test that the original problem does not contain a trivially unsatisfiable tableau row: one where the bounds obtained by interval arithmetic on the right-hand side of $b = \sum_n t_{b,n} n$ the equality have an empty intersection with those for $b$.
\item During the forced pivoting procedure, we detect whether the new equality obtained is trivially unsatisfiable, in which case we terminate immediately.
\item Forced pivots can be done in any order. At the beginning of the procedure, we sort the variables in $B_i \setminus B_f$ according to the number of nonzero coefficients in the corresponding equation. When choosing the basic variable to be pivoted, we take the least one in that ordering. The idea is that the pivoting steps are faster when the equation defining the basic variable to be pivoted has few nonzero coefficients.
\item Similarly, one can pre-sort the variables in $B_f \setminus B_i$ according to their number of occurrences in equations.
\end{itemize}

The SMT procedure may have at its disposal, in addition to the ``exact'' theory test a partial and  ``fast'' theory test, which may err on the side of satisfiability: first test satisfiability in floating-point, and test in exact arithmetic only if negative. The ``fast'' theory test may be used to test whether it seems a good idea to investigate a branch of the search tree or to backtrack, while the ``exact'' results should be used for theory propagation or when it seems a good idea to check whether the current branch truly is satisfiable. Various systems are possible, depending on how the SAT and theory parts interact in the solver~\cite{Faure_et_al_SAT08}.

Note that after an ``exact'' theory check, we have an exact rational simplex tableau corresponding to the constraints, from which it is possible to extract theory propagation information. For instance, if interval analysis on a row $x = 3y+5z$, using intervals from $y$ and $z$, shows that $x < 3$, then one can immediately conclude that in the current SAT search branch, the literal $x \geq 4$ is false.

\section{Implementation and benchmarks}
We implemented the above algorithm into a toy SMT solver.%
\footnote{Benchmarks and implementation are available from\\
\url{http://www-verimag.imag.fr/~monniaux/simplexe/}.}
The SAT part is handled by \soft{Minisat}~\cite{Minisat_SAT03}. Implementing a full SMT solver for LRA was useful for testing the software, against the SMT-LIB examples.%
\footnote{\url{http://goedel.cs.uiowa.edu/smtlib/benchmarks/QF_LRA.tar.gz}}
We however did not use SMT-LIB for benchmarking: the performance of a complete SMT solver is influenced by many factors, including the performance of the SAT solver, the ease of adding new clauses on the fly, etc., outside of the pure speed of the decision procedure.

The floating-point simplex used is the dual simplex implemented by option \verb+GLP_DUALP+ in \soft{Glpk}~\cite{GLPK}, with default parameters.
The rational simplex tableau is implemented using sparse vectors.%
\footnote{More precisely, using \texttt{boost::numeric::ublas::compressed\_vector} from the \soft{Boost} library, available at \url{http://www.boost.org/}.}
Rational numbers are implemented as quotients of two 32-bit numbers; in case of overflow, extended precision rationals from the GMP library \cite{GMP} are used.%
\footnote{\url{http://gmplib.org/}} The reason behind a two-tier system is that GMP rationals inherently involve some inefficiencies, including dynamic memory allocation, even for small numerator and denominator. In many examples arising from verification problems, one never needs to call GMP. The simplex algorithm, including the pivot strategy, is implemented straight from \cite{SRI-CSL-06-01,DBLP:conf/cav/DutertreM06}. It is therefore likely that our system can be implemented as a preprocessor into any of the current SMT platforms for linear real arithmetic, or even those for linear integer arithmetic, since these are based on relaxations to rationals with additional constraints (branch-and-bound or Gomory cuts).

We benchmarked four tools:
\begin{itemize}
\item Our ``naive'' implementation of rational simplex.
\item The same, but with the floating-point preprocessing and forced pivoting phase described in this article (``mixed simplex'').
\item Bruno Dutertre and Leonardo de Moura's (SRI International) Yices 1.0.9
\item Nikolaj Bj{\o}rner and Leonardo de Moura's (Microsoft Research) Z3,
  as presented at SMT-COMP~'08.
\end{itemize}

We used two groups of benchmarks:
\begin{itemize}
\item Benchmarks kindly provided to us by Leonardo de Moura, extracted from SMT-LIB problems. Each benchmark is an unsatisfiable conjunction, used by Z3 as a theory lemma. These examples are fairly sparse, with small coefficients, and rational simplex seldom needs to use extended precision rational arithmetic. Despite this, the performance of the ``mixed'' implementation is better than that of Yices~1 (Fig.~\ref{fig:benchmarks_Leonardo}). Since the source code of neither Yices nor Z3 are available, the reasons why Z3 performs better than Yices~1 and both perform much better than our own rational simplex implementation are somewhat a matter of conjecture. The differences probably arise from both a better sparse matrix implementation, and a better pivoting strategy.

\item Random, dense benchmarks. On these, our mixed implementation performs faster than all others, including the state-of-the-art SMT solvers (Fig.~\ref{fig:benchmarks_random_50_100_100}).
\end{itemize}

On a few examples, \soft{Glpk} crashed due to an internal error (failed assertion). We are unsure whether this is due to a bug inside this library or our misusing it --- in either case, this is rather unsurprising given the complexity of current numerical packages. It is also possible that the numerical phase outputs incorrect results in some cases, because of bugs or loss of precision in floating-point computations. Yet, this has no importance --- the output of the numerical phase does not affect the correction of the final result, but only the time that it takes to reach this result.

\section{Related work}
The high cost of exact-arithmetic in linear programming has long been recognized, and linear programming packages geared towards operational research applications seldom feature the option to perform computations in exact arithmetic. In contrast, most SMT solvers (e.g. \soft{Yices}, \soft{Z3}) or computational geometry packages implement exact arithmetic.

Faure et al. have experimented with using commercial, floating-point SMT solvers such as CPLEX%
\footnote{CPLEX is a professional optimization package geared towards large operational research problem, published by ILOG. \url{http://www.ilog.com/products/cplex/}}
inside an SMT solver~\cite{Faure_et_al_SAT08}. Their approach is different from ours in that they simply sought to reuse the yes/no return value produced by the inexact solver, while we also reuse the basis structure that it produces. Many of the difficulties they report --- for instance, not being able to reuse the output of the inexact solver for theory propagation --- disappear with our system. Some of their remarks still apply: for instance, the floating-point solver should allow incremental use, which means that, for instance, it should not perform LU matrix factorization every time a bound is changed but only at basis initialization.

The idea of combining exact and floating-point arithmetic for the simplex algorithm is not new. G\"artner proposed an algorithm where most of the computations inside the simplex algorithm are performed using ordinary floating-point, but some are performed using extended precision arithmetic (floating-point with an arbitrary number of digits in the mantissa), and reports improvements in precision compared to the industrial solver CPLEX at moderate costs~\cite{Gaertner_SODA98}. It is however difficult to compare his algorithm to ours, because they are geared towards different kinds of problems. Our algorithm is geared towards the decision problem, and is meant to be potentially incremental, and to output both satisfiability and unsatisfiability witnesses in exact rational arithmetic. G\"artner's is geared towards optimization problems in computational geometry, and does not provide optimality or unsatisfiability witnesses.

The work that seems closest to ours seems to be LPex~\cite{Dihflaoui_et_al_SODA03}, a linear programming solver that first obtains a floating-point solution, then recreates the solution basis inside an exact precision solver and then verifies the solution and possibly does additional iterations in exact precision in order to repair a ``wrong'' basis. The original intent of the authors was to apply this technique to optimization, not decision problems, and there are in fact arguments against using this technique for optimization that do not apply to using it for decision. A simplex-based optimization solver actually runs two optimization phases:
\begin{enumerate}
\item A search for a feasible solution (e.g. the algorithm from Sec.~\ref{part:basic_simplex}), which can be equivalently presented as an optimization problem. A new variable $\delta$ is added, and all inequalities $\sum a_{i,j} x_j \leq b$ are replaced by $\sum a_{i,j} x_j - \delta \leq b$. By taking $\delta$ sufficiently large, one can find a solution point of the modified problem. Optimization iterations then reduce $\delta$ until it is nonnegative. Intuitively, $\delta$ measures how much afar we are from a solution to the original problem.

\item Once a solution point to the original problem is found, the original objective function $f$ is optimized.
\end{enumerate}
There are therefore two objective functions at work: one for finding a solution point, and the ``true'' objective function. The floating-point simplex process, optimizes $\delta$, then~$f$. Then, the rational simplex, seeking to ``repair'' the resulting basis, starts optimizing $\delta$ again; then it has to move back to optimizing~$f$. In general, changing objective functions in the middle of an optimization process is bad for efficiency. However, since we are interested in decision, we optimize a single function and the objection does not hold.

This approach was later improved to computing exact solutions for all problems in NETLIB-LP, a popular benchmark library for linear programming problems~\cite{Koch_thefinal_ORL04}. One improvement was to re-run the floating-point simplex in higher precision rather than attempt to repair the ``wrong'' output basis using exact arithmetic --- thus the suggestion to use extended-precision floating-point arithmetic%
\footnote{Available through GNU MP's \texttt{mpf} type or the MPFR library, for instance.}
with increasing precisions until the solution found can be checked in exact precision. This last algorithm was implemented inside the QSopt\_ex solver \cite{DBLP:journals/orl/ApplegateCDE07}.%
\footnote{\url{http://www2.isye.gatech.edu/~wcook/qsopt/index.html}}

These proposals are however different from ours in that they involve more modifications to the underlying exact arithmetic solver. For once, they compute exact precision factorizations of the input matrices, which probably saves time for the large operational research problems that they consider but may have a too high overhead for the smaller problems that arise in SMT applications. In contrast, our algorithm can be adapted as a preprossessing step to the simplex procedure used in existing SMT solvers with hardly any modification to that procedure.

\section{Conclusion and future work}
Our work leverages the maturity and efficiency of floating-point simplex solver (inherent efficiency of hardware floating-point versus rationals, advanced pricing and pivoting strategies...) in order to speed up exact decision procedures in cases where these perform poorly.

The main application of SMT solvers so far has been program or specification verification. Such problems are typically sparse, with small coefficients. On such problems, recent SMT solvers such as Yices and Z3 typically perform well using the simplex algorithm over rational numbers. Performance, however, decreases considerably if they are used over dense problems, since the size of numerators and denominators involved can become prohibitively large. In this case, running our preprocessing phase before embarking on costly extended precision simplex can save significant time. We suggest that our procedure be added to such implementations and activated, as a heuristic, when the rational coefficients become too large.

Allowing SMT solvers to scale beyond program analysis examples may prove useful if they are used for some other applications than program proofs, for instance, formal mathematical proofs. As an example of the use of formal arithmetic proofs outside of program verification, Hales proved Kepler's conjecture using many lemmas obtained by optimization techniques~\cite{Hales05_Kepler}, but mathematicians objected that it was unclear whether these lemmas truly held. As a result, Hales launched a project to formally prove his theorem, including all lemmas obtained using numerical optimization. He proposed transformations of his original linear programming problems into problems for which it is possible to turn the bounds obtained by numerical techniques into numerical bounds~\cite{Hales02_Kepler}. This suggests that there are potential mathematical applications of efficient decision procedures for linear arithmetic.

The applications of our improvements are not limited to the linear theory of the reals or rationals. They also apply to the linear theory of integers, or the mixed linear theory of rationals/reals and integers. In most SMT solvers, decision for integer and mixed linear problems is implemented by relaxing the problem to the real case. If there is no real solution, then there is no integer solution; if there is a solution where the variables that are supposed to be integers are integers, then the integer or mixed problem has a solution. Otherwise, the search is restarted in parts of the state space, delimited by ``branch and bound'' or Gomory cuts. Any efficiency improvement in the rational simplex can thus translate into an improvement to the integer or mixed linear decision procedure.

The reason why we still have to perform expensive rational pivots even after computing a floating-point solution is that the floating-point solutions produced by the simplex algorithm almost always lie outside of the solution polyhedron when tested over exact arithmetic, as explained in Sec.~\ref{part:mixed}. We therefore think of investigating interior point methods, searching both for a satisfiability witness and for an unsatisfiability witness.

\bibliographystyle{plain}
\bibliography{float_simplex}
\end{document}